\documentclass[runningheads]{llncs}

\usepackage{graphicx}
\usepackage{xcolor} 
\usepackage{microtype}
\usepackage{xypic}
\usepackage{nicefrac}
\usepackage{microtype,comment}
\usepackage{enumerate}
\usepackage{boxedminipage}
\usepackage{tikz}
\usepackage{hyperref}
\usepackage{amssymb,amsmath}
\usepackage{slashbox}
\usepackage[T1]{fontenc}
\usepackage{graphicx}
\input
\usetikzlibrary{arrows}
\usetikzlibrary{shapes}

\usetikzlibrary{shapes}
\usepackage{comment}

\tikzset{
  vert/.style={circle, draw=black!100,fill=black!100,thick, inner sep=0pt, minimum size=.8mm},
  empty/.style={draw=none, fill=none, minimum size=0mm, inner sep=0pt}}

\newcommand{\NP}{{\sf NP}}
\newcommand{\W}{{\sf W}}
\newcommand{\cP}{{\sf P}}

\newcommand{\QCSP}{\mathrm{QCSP}}

\newcounter{ctrclaim}[theorem]
\newcounter{ctrcase}[theorem]

\newcommand{\FPT}{{\sf FPT}}

\newcommand{\problemdef}[3]{
        \begin{center}
                \begin{boxedminipage}{.99\textwidth}
                        \textsc{{#1}}\\[2pt]
                        \begin{tabular}{ r p{0.8\textwidth}}
                                \textit{~~~~Instance:} & {#2}\\
                                \textit{Question:} & {#3}
                        \end{tabular}
                \end{boxedminipage}
        \end{center}
}

\begin{document}

\title{Graph Homomorphism, Monotone Classes and Bounded Pathwidth\thanks{\textcolor{black}{An extended abstract of this article appeared at the 20th Conference on Computability in Europe (CiE 2024) \cite{CiE2024}.}}}
\titlerunning{Graph Homomorphism, Monotone Classes and Bounded Pathwidth}
\author{Tala Eagling-Vose\inst{1}\orcidID{0009-0008-0346-7032} \and
Barnaby Martin\inst{1}\orcidID{0000-0002-4642-8614} \and
Dani\"el Paulusma\inst{1}\orcidID{0000-0001-5945-9287} \and
Siani Smith \inst{2}\orcidID{0000-0003-0797-0512
}
}

\authorrunning{T. Eagling-Vose, B. Martin, D. Paulusma, S. Smith}

\institute{Durham University, Durham, United Kingdom \email{\{tala.j.eagling-vose,barnaby.d.martin,daniel.paulusma\}@durham.ac.uk} \and
University of Bristol and Heilbronn Institute for Mathematical Research, Bristol, United Kingdom
\email{siani.smith@bristol.ac.uk}
}

\maketitle             

\begin{abstract}
\textcolor{black}{
In recent work by Johnson et al. (2022), a framework was described for the study of graph problems over classes specified by omitting each of a finite set of graphs as subgraphs. If a problem falls into the framework then its computational complexity can be described for all such graph classes, giving a dichotomy between those classes for which the problem is hard and those for which it is easy. Usually, hard is NP-complete and easy is in P, though some examples are given where hard is quadratic and easy is almost linear.}

\textcolor{black}{In this article}, we consider several variants of the homomorphism problem in relation to this framework. It is known that certain homomorphism problems, e.g. {\sc $C_5$-Colouring}, do not sit in the framework. By contrast, we show that the more general problem of {\sc Graph Homomorphism} does sit in the framework, \textcolor{black}{with hard cases NP-complete and easy cases in P}. 

\textcolor{black}{We go on to} consider several locally constrained variants of the homomorphism problem, namely the locally bijective, surjective and injective variants. Like {\sc $C_5$-Colouring}, none of these is in the framework. However, where a bounded-degree restrictions are considered, we prove that each of these problems is in our framework, \textcolor{black}{with hard cases NP-complete and easy cases in P}.

\textcolor{black}{Next}, we give the first example of a problem in the framework such that hardness is in the polynomial hierarchy above NP. This comes from a list colouring game, \textcolor{black}{realised through first-order logic as quantified constraints}. We show that with the \textcolor{black}{additional} restriction of bounded alternation, the problem is contained in the framework. The hard cases are $\Pi_{2k}^\mathrm{P}$-complete and the easy cases are in P.

\textcolor{black}{Finally, we go on to consider an aforementioned problem from our framework, complete for the second level of the polynomial hierarchy, under the omission in the input of not just a graph, but rather a graph $H$ annotated with the types for each vertex: existential or universal. This natural formulation gives finer control over the structure that we forbid in the input sentences, which we realise through the annotated graphs. We address especially the case where $|V(H)|$ has size $3$.}

\end{abstract}

\section{Introduction}

Let $G$ and $H$ be two graphs. If $H$ can be obtained from $G$ by a sequence of vertex deletions only, then $H$ is {\it an induced} subgraph of $G$; else $G$ is {\it $H$-free}. The induced subgraph relation has been well studied in the literature for many classical graph problems, such as {\sc Colouring} \cite{KKTW01}, {\sc Feedback Vertex Set} \cite{PPR22}, {\sc Independent Set} \cite{GKPP22}, and so on.

Here we focus on the subgraph relation. A graph $G$ is said to contain a graph $H$ as a {\it subgraph} if $H$ can be obtained from $G$ by a sequence of vertex deletions and edge deletions; else $G$ is said to be {\it $H$-subgraph-free}. For a set of graphs ${\cal H}$, a graph $G$ is {\it ${\cal H}$-subgraph-free} if $G$ is $H$-subgraph-free for every $H\in {\cal H}$; we also write that $G$ is $(H_1,\ldots,H_p)$-subgraph-free, if ${\cal H}=\{H_1,\ldots,H_p\}$. 
Graph classes closed under edge deletion are also called {\it monotone}~\cite{ABKL07,BL02}. Monotone classes that are specified by a finite set of omitted subgraphs are called 
{\it finitely-bounded}.

When compared to those for $H$-free graphs there are relatively few complexity classifications for $H$-subgraph-free graphs. Despite this; see~\cite{AK90} for complexity classifications of {\sc Independent Set}, {\sc Dominating Set} and {\sc Longest Path};  
and~\cite{GP14,Ka12} for classifications of {\sc List Colouring} and {\sc Max Cut}, respectively. 
All of these classifications hold even for ${\cal H}$-subgraph-free graphs, where ${\cal H}$ is any finite set of graphs. 
In general, such classifications might be hard to obtain; see, for example, \cite{GPR15} for a partial classification of {\sc Colouring} for $H$-subgraph-free graphs. 

Therefore, in~\cite{JMOPPSV}
a more systematic approach was followed, namely by introducing a new framework for ${\cal H}$-subgraph-free graph classes (finite ${\cal H}$) adapting the approach 
of~\cite{Ka12}. 
To explain the framework of~\cite{JMOPPSV} we need to introduce some additional terminology. {\it Treewidth} and {\it pathwdith} are two widly studied width parameters broadly capturing likeness to a tree or path respectivly.
A class of graphs has bounded {\it treewidth} or {\it pathwdith} if there exists a constant~$c$ such that every graph in it has treewidth or pathwidh, respectively, at most~$c$. Now let $G=(V,E)$ be a graph.
Then $G$ is {\it subcubic} if every vertex of $G$ has degree at most~$3$.
The {\it subdivision} of an edge $e=uv$ of $G$ replaces $e$ by a new vertex $w$ with edges $uw$ and $wv$. For an integer $k\geq 1$, the {\it $k$-subdivision} of~$G$ is the 
graph~$G^k$ obtained from~$G$ by subdividing each edge of $G$ exactly $k$ times. For a class of graphs ${\cal G}$ and an integer~$k$, ${\cal G}^k$ consists of the $k$-subdivisions of the graphs in ${\cal G}$.

We say that a graph problem $\Pi$ is computationally hard
{\it under edge subdivision of subcubic graphs} if for every integer $j \geq 1$ there is an integer~$\ell \geq j$ such that:
if $\Pi$ is computationally hard for the class ${\cal G}$ of subcubic graphs, then $\Pi$ is computationally hard for ${\cal G}^{\ell}$. The framework of~\cite{JMOPPSV} makes a distinction between ``efficiently solvable'' and ``computationally hard'', which could for example mean a distinction between $\cP$ and \NP-complete. 

Commonly, we can prove the condition holds by showing that computational hardness is maintained under $k$-subdivision for a small integer $k$ (e.g.~$k=1,2,3,4$) and then repeatedly apply the $k$-subdivision operation. 
%
We can therefore say a graph problem~$\Pi$ is a {\it C123-problem} (belongs to the framework) if it satisfies the 
three conditions:

\begin{enumerate}
\item [{\bf C1.}] $\Pi$ is efficiently solvable for every graph class of bounded pathwidth\footnote{In the original framework paper \cite{JMOPPSV}, pathwidth and treewidth are interchangeable in this position.};
\item [{\bf C2.}] $\Pi$ is computationally hard 
 for the class of subcubic graphs; and
\item [{\bf C3.}] $\Pi$ is computationally hard under edge subdivision 
of subcubic graphs.
\end{enumerate}

\noindent
To describe the impact of these conditions, we need some notation. The {\it claw} is the $4$-vertex star.  A {\it subdivided} claw is a graph obtained from a claw after subdividing each of its edges zero or more times. The {\it disjoint union} of two vertex-disjoint graphs $G_1$ and $G_2$ has vertex set $V(G_1)\cup V(G_2)$ and edge set  $E(G_1)\cup E(G_2)$. The set ${\cal S}$ consists of the graphs that are disjoint unions of subdivided claws and paths. As shown in~\cite{JMOPPSV}, C123-problems allow for full complexity classifications for ${\cal H}$-subgraph-free graphs (as long as ${\cal H}$ has finite size).

\begin{theorem}[\cite{JMOPPSV}]\label{t-dicho2}
Let $\Pi$ be a C123-problem.
For a finite set ${\cal H}$, the problem $\Pi$ on ${\cal H}$-subgraph-free graphs is efficiently solvable if ${\cal H}$ contains a graph from ${\cal S}$ and computationally hard otherwise.
\end{theorem}

\noindent
Examples of C123-problems include {\sc Independent Set}, {\sc Dominating Set}, {\sc List Colouring}, {\sc Odd Cycle Transversal}, {\sc Max Cut}, {\sc Steiner Tree} and {\sc Vertex Cover}; see~\cite{JMOPPSV}. However, there are still many graph problems that are not C123-problems, such as {\sc Colouring} (whose classification is still open even for $H$-subgraph-free graphs). Hence, it is a natural question if those problems can still be classified for graph classes defined by some set of forbidden subgraphs.

This paper considers the {\sc Graph Homomorphism} problem alongside two variants: a graph colouring game, and locally constrained homomorphism problems. Here, we will define the first of these, while details of the latter two will be addressed in greater detail in their respective sections.

The problem {\sc $H$-Colouring} takes an input graph $G$ and asks whether there is a homomorphism from $G$ to $H$ i.e a function $h : V(G) \to V(H)$ such that $xy \in E(G)$ then $h(x)h(y) \in E(H)$. If there is, we write $G \rightarrow H$. The more general problem {\sc Graph Homomorphism} takes both $G$ and $H$ as input, with the same question, whether $G \rightarrow H$. In general, {\sc $H$-Colouring} is not a C123-problem, for example, {\sc $C_5$-Colouring} does not satisfy C3 (as it satisfies C1 and C2, we label it a {\it C12-problem}). The topic of such C12-problems is elaborated in \cite{MPPSMvL}. By contrast, we show that {\sc Graph Homomorphism} is a C123-problem (where C3 is applied uniformly to both $G$ and $H$). Grohe has argued in \cite{Gr07} that, assuming $\FPT\neq \W[1]$ (an assumption widely believed in Parameterized Complexity), {\sc Graph Homomorphism} is in P on a restricted class of graphs if, and only if that class has bounded treewidth. 
Bearing in mind the interchangeability of treewidth and pathwidth in our framework, we are improving ``not in P'' to ``NP-complete'' for finitely-bounded monotone classes.

We continue by considering locally constrained variants of the {\sc Graph Homomorphism} problem. These problems, {\sc Locally Bijective Homomorphism}, {\sc Locally Surjective Homomorphism} and {\sc Locally Injective Homomorphism}, denoted {\sc LBHom}, {\sc LSHom} and {\sc LIHom} respectively, have all been intensively studied in the literature~\cite{FK08}.
When either $G$ or $H$ has bounded (vertex) degree, we prove that each of these three problems is a C123-problem, just like {\sc Graph Homomorphism}, with a dichotomy between P and NP-complete on finitely-bounded monotone classes. Whereas without such a degree bound they are C23-problems, i.e satisfy only C2 and C3 and are elaborated further in \cite{BJMOPPSV}.

The graph colouring game we consider is a variant of that proposed by Bodlaender in \cite{Bo91}, {\sc Sequential Colouring Construction Game}. Here two players alternate in colouring vertices of a graph with one of $k$ colours, with the first player winning when a proper colouring is obtained. In particular, fixing $k=3$, this problem closely relates to the \emph{quantified constraint satisfaction problem} {\sc QCSP$(K_3)$}. Both of these problems are PSPACE-complete \cite{Bo91,BBCJK09}.

Bodlaender proves that {\sc Sequential $3$-Colouring Construction Game} is in P on any class of bounded vertex separation number. It is well-known that classes of bounded vertex separation number and classes of bounded pathwidth coincide \cite{Ki92}. However here the vertex separation number is with respect to the order in which the vertices are played. A na\"ive reading of \cite{Bo91} risks a fundamental misinterpretation, as we prove that {\sc Sequential $3$-Colouring Construction Game} is PSPACE-complete on some class of bounded pathwidth. We do this by mediating through the closely related problem {\sc QCSP$(K_3)$} using a celebrated result of Atserias and Oliva \cite{AO14}. In that paper, they prove that {\sc Quantified Boolean Formulas} (QBF) is PSPACE-complete on some class of bounded pathwidth, even when the input is restricted to conjunctive normal form (CNF).

Within the framework, just as {\sc List Colouring} was considered, we can extend our attention to a list colouring game. Adjoining some unary relations to {\sc $\QCSP(K_3)$}, is analogous to moving from {\sc $3$-Colouring} to {\sc List $3$-Colouring}. Using only two of these lists we arrive at the problem $\QCSP(K_3, \{1,2\},\{1,3\})$. Owing to the aforementioned hardness for bounded pathwidth, and unlike {\sc List Colouring}, $\QCSP(K_3, \{1,2\},\{1,3\})$ is not a C123-problem, but is a C23-problem. 

However, considering the bounded alternation restriction of this, we prove that $\Pi_{2k}$-$\QCSP(K_3, \{1,2\},\{1,3\})$ is a C123-problem, for which the hard cases are $\Pi^{\mathrm{P}}_{2k}$-complete while the easy cases are in P.

\textcolor{black}{The significance of our work lies in its uncovering the field of homomorphism problems and their variants as rich ground for finding framework problems, notwithstanding that is was known that $H$-colouring problems are not in the framework in general. Indeed, we go on to consider the problem $\Pi_{2}$-$\QCSP(K_3, \{1,2\},\{1,3\})$ under restrictions that concern omitting as a subgraph not just a graph but rather a graph in which all the vertices are labelled either universal or existential. We argue that this allows for finer classifications, and we classify the complexity for all such labelled graphs of size at most $3$, while proving also some general results.}

\textcolor{black}{We conclude by considering} our framework in relation to bounded pathwidth, which coincides (on finitely-bounded monotone classes) with essentially omitting some subdivided claw 
as a subgraph~\cite{RS84}. If we instead omit some path, we essentially get the graphs of bounded treedepth~\cite{NOM12}. 
For {\sc LBHom}, {\sc LSHom} and {\sc LIHom}, there exists some bounded treedepth class on which they are \NP-complete~\cite{BDKOP22}. We will describe a C23-problem, namely {\sc Long Edge Disjoint Paths} that is easy on graphs of bounded treedepth but hard on certain classes of bounded pathwidth.  Within previous work on C23-problems an example with this behaviour has not yet been discussed. 
The complexity of QBF on bounded treedepth is a famous open problem (see e.g.~\cite{FGH23}, where it is proved to be in P under some further restrictions). The complexity of $\QCSP(K_3)$ on bounded treedepth may be similarly elusive to classify, and it is on this note that we conclude.

\textcolor{black}{This paper is organised as follows. After the preliminaries, we address graph homomorphism in Section~\ref{sec:graph-homomorphism} and locally constrained graph homomorphism in Section~\ref{sec:graph-local-homomorphism}. We then consider sequential colouring games in Section~\ref{s-seq} and QCSPs in Section~\ref{s-seq-plus}. We then go on to consider bounded pathwidth in Section~\ref{sec:long} and labelled graphs in Section~\ref{sec:new-for-journal}. We conclude with some final remarks and open question in Section~\ref{sec:final-remarks}. We sum up our results in Figure~\ref{fig:results}.}

{\color{black}
This paper is a considerable extension of \cite{CiE2024}. The proofs have been elaborated and Figure~\ref{fig:results} has been added. The whole of Section~\ref{sec:new-for-journal} is new.
}

\begin{figure}[h]
\color{black}
\begin{center}
\begin{tabular}{|c|c|c|c|}
\hline
problem & C123-framework & easy & hard \\
\hline
{\sc $C_5$-Colouring} & no & & \\
\hline
{\sc Graph Homomorphism} & yes & P & NP-complete \\
\hline
  {\sc LBHom} & no & & \\
  {\sc LSHom} & no & & \\
    {\sc LIHom} & no & & \\
\hline
  {\sc Degree-$3$-LBHom} & yes & P & NP-complete \\
  {\sc Degree-$3$-LSHom} & yes & P & NP-complete \\
    {\sc Degree-$3$-LIHom} & yes & P & NP-complete \\
\hline
$\Pi_{2k}\mbox{-}\mbox{QCSP}(K_3, \{1,2\},\{1,3\})$ & yes & P & $\Pi_{2k}^{\mathbf{P}}\mbox{-}$complete \\
\hline
\end{tabular}
\end{center}
\caption{Summary of our results.}
\label{fig:results}
\end{figure}

\section{Preliminaries}

A \textit{tree decomposition} for a graph $G=(V,E)$ is a pair $(T,X)$ where $T$ is a tree and $X$ consists of subsets of vertices from $V$ which we call bags. Each node of $T$ corresponds to a single bag of $X$. For each vertex $v \in V$ the nodes of $T$ containing $v$ must induce a non-empty connected subgraph of $T$ and for each edge $uv \in E$, there must be at least one bag containing both $u$ and $v$. Similarly, we can define a \textit{path decomposition} where $T$ must instead be a path. We can then define the \textit{width} of $(T,X)$ to be one less than the size of the largest bag. From this, the \textit{treewidth} of a graph, $\mathit{tw}(G)$, is the minimum width of any \textit{tree decomposition} and the \textit{pathwidth} $\mathit{pw}(G)$, is the minimum width of any \textit{path decomposition}. As every path decomposition is also a tree decomposition $\mathit{tw}(G) \leq \mathit{pw}(G)$.

In the following, $G$ and $H$ are graphs, and $f$ is a {\it homomorphism} from $G$ to $H$, that is, $f(u)f(v)$ is an edge in $H$ whenever $uv$ is an edge in $G$. We denote the (open) neighbours of a vertex $u$ in $G$ by $N_G(u)=\{v\; |\; uv\in E(G)\}$. We say that $f$ is {\it locally injective, locally bijective} or {\it locally surjective for a vertex $u \in V(G)$} if the restriction $f_u: N_G(u)\rightarrow N_H(f(u))$ of $f$ is injective, bijective or surjective, respectively. Now,
$f$ is said to be {\it locally injective}, {\it locally bijective}  or {\it locally surjective} if $f$ is locally injective, locally bijective or locally surjective for every $u \in V(G)$.
{\color{black}
If $X$ is some subset of vertices $X \subseteq V(G)$, then $G[X]$ is the subgraph of $G$ induced by $X$. We define the closed neighbourhood $N[X]$ to be $\{v\; |\; v \in X \vee (\exists u \in X \wedge uv\in E(G))\}$.
}

$\QCSP(\mathcal{B})$, is defined for some relational structure $\mathcal{B}$ which for us will always be a graph. The problem takes as input a sentence $\phi = Q_1x_1Q_2x_2 \cdot Q_nx_n \Phi$ such that $Q_i \in \{ \exists,\forall \}$ and $\Phi$ is a conjunction of atomic formulas, \textit{constraints}. The primal graph of a formula $\Phi$ contains a vertex for each $x_i \in \Phi$ and an edge if and only if the two variables occur together in a constraint. Let $\Pi_{2k}$-$\QCSP(\mathcal{B})$ be the restriction of this problem to sentences in $\Pi_{2k}$-form, \mbox{i.e.} with quantifier prefix leading with universal quantifiers, alternating $2k-1$ times, concluding with existential quantifiers.

\section{Graph Homomorphism}
\label{sec:graph-homomorphism}

In this section we will prove that {\sc Graph Homomorphism} is a C123-problem.

\problemdef{Graph Homomorphism}{A graph $G$ and a graph $H$}{Is there a homomorphism from $G$ to $H$?}

Let us recall that we apply C1, C2 and C3 to both graphs. In particular, the subdivision of C3 uniformly applies to both inputs $G$ and $H$. For a graph $G$, recall that $G^r$ is $G$ with each edge replaced by a path of length $r+1$ (an $r$-subdivision). For example, ${K_3}^5=C_{15}$.

\begin{lemma}[\cite{DP89}]\label{l-dp89}
    {\sc Graph Homomorphism} satisfies C1.
\end{lemma}
\begin{lemma}[\cite{GHN00}]
    {\sc Graph Homomorphism} satisfies C2.
\end{lemma}
Let us recall that we apply the subdivision of C3 uniformly to both inputs $G$ and $H$.
\begin{lemma} \label{lem:graph-hom-weak-C3}
    {\sc Graph Homomorphism} satisfies the variant of C3 that does not restrict to subcubic.
\end{lemma}
\begin{proof}
    Let $r\geq 1$ be fixed. Let $G$ be connected and have the additional property that every edge is in a triangle. It is clear that this subset remains \NP-complete, \textcolor{black}{since one can enforce it by adding a new triangle at each edge.} We claim that:
    \[ G \rightarrow K_{3}  \mbox{ if and only if } G^{5^{r}-1} \rightarrow C_{3 \cdot 5^{r}}.\]
    \textcolor{black}{An example of this reduction appears in Figure~\ref{fig:new-figure-proof-one}.}
    The forward direction is trivial. Let us address the backward direction and let $h$ be a homomorphism from $G^{5^r-1}$ to $C_{3\cdot 5^{r}}$. Let $X$ be the set of vertices of $G^{5^r}$ that appear already in $G$. We claim that $h(X) \subseteq \{i,i+5^r,i+2\cdot 5^r\}$ for some $i \in [3\cdot 5^{r}]$ where addition is $\bmod\ 3 \cdot 5^{r}$. Suppose otherwise, then there is an edge $xy$ in $G$, so that the distance between $h(x)$ and $h(y)$ in $C_{3 \cdot 5^{r}}$ is $0<i <5^r$ (this is why we assumed $G$ to be connected). But this is impossible since we may consider there exists $z$ so that $xyz$ is a triangle in $G$ and this triangle must have been mapped to line in $C_{3 \cdot 5^{r}}$. Now, once we calculate $i$ and subtract it, we can divide by $5^r$ to use $h$ restricted to $X$ to give a homomorphism from $G$ to $K_3$.
\qed
\end{proof}
\begin{figure}
{\color{black}
\begin{center}
\resizebox{!}{0.2cm}{
$
\xymatrix{
\\
\bullet  \ar@{-}[rr] \ar@{-}[dd] & & \bullet \ar@{-}[dd] \\ 
\\
\bullet \ar@{-}[rr] & & \bullet \\
\\
}
$
\hspace{1cm}
$
\xymatrix{
& & \bullet \ar@{-}[dr] \ar@{-}[dl]\\
& \bullet  \ar@{-}[rr] \ar@{-}[dd] & & \bullet \ar@{-}[dd] \\ 
\bullet \ar@{-}[ur] \ar@{-}[dr] & & & & \bullet \ar@{-}[ul] \ar@{-}[dl] \\
& \bullet \ar@{-}[rr] & & \bullet \\
& & \bullet \ar@{-}[ur] \ar@{-}[ul] \\
}
$
}
\end{center}
\caption{Reduction from the proof of Lemma~\ref{lem:graph-hom-weak-C3}. This is what the $4$-cycle $C_4$ would be reduced to. The first step is to assume each edge is in a triangle. Finally, the edges will be substituted by paths of length $5^r$.}
\label{fig:new-figure-proof-one}
}
\end{figure}

\noindent

Note that 
Lemmas~\ref{l-dp89} and~\ref{lem:graph-hom-weak-C3} are 
enough to guarantee a dichotomy for $H$-subgraph-free graphs where $H$ is a single graph and not a finite set of graphs \cite{JMOPPSV}. 
However, as $C_3$-{\sc Colouring} (or equivalently, $3$-{\sc Colouring}) does not satisfy C2 due to 
Brooks' Theorem, we need to do more work to
accomplish the following.

\begin{lemma}
    {\sc Graph Homomorphism} satisfies C3.
    \label{lem:graph-hom-weak-C3-bis}
\end{lemma}
\begin{proof}
Let us recall the self-reduction from 
{\sc $C_5$-Colouring} to (subcubic) {\sc $C_5$-Colouring} in Theorem~3.1 from \cite{GHN00}. 
Each vertex with degree $d$ becomes a chain of $d$ $C_5$s, the $i+1$th connected to the $i$th by identifying edge $1$ on the former with edge $4$ on the latter (one may take any cyclic ordering of the edges). Then the $i$th occurrence of the vertex is taken to be the top vertex of the $i$th $C_5$ in the chain (where edge $2$ meets edge $3$). Now, one can simply join the $i$th occurrence of vertex $x$ to the $j$th occurrence of vertex $y$ if the edge $xy$ is the $i$th edge of $x$ and the $j$th edge of $y$. This is exactly the reduction of Theorem 3.1 from \cite{GHN00}. We amend if by pretending each edge $xy$ is in fact three edges and must be joined from the chain of $C_5$s representing $x$ to the chain of $C_5$s representing $y$ three times (this can at most result in a tripling of the degree). We assume that these joins are consecutive on each chain. \textcolor{black}{An example of this self-reduction appears in Figure~\ref{fig:new-figure-proof-bis}.}

Let $Y$ be the set of instances of 
{\sc $C_5$-Colouring} that can be obtained by this self-reduction. For $G \in Y$, all edges are in a $C_5$ except perhaps the edges that came from the edges in the original graph, which are now represented in $G$ in triplicate.
    Let $r\geq 1$ be fixed. Let $G \in Y$. We claim that:
    \[G \rightarrow C_{5}  \mbox{ if and only if } G^{5^r-1} \rightarrow C_{5^{r+1}}.\]
    The forward direction is trivial. Let us address the backward direction and let $h$ be a homomorphism from $G^{5^r}$ to $C_{5^{r+1}}$. Let $X$ be the set of vertices of $G^{5^r}$ that appear already in $G$. We claim that $h(X) \subseteq \{i,i+5^r,i+2\cdot 5^r,i+3\cdot 5^r,i+4\cdot 5^r\}$ for some $i \in [5^{r+1}]$ where addition is $\bmod\ 5^{r+1}$. Suppose otherwise, then there is an edge $xy$ in $G$ so that the distance between $h(x)$ and $h(y)$ in $C_{5^{r+1}}$ is $0<i <5^r$. But this is impossible for the edges in $G$ that were in a $C_5$ since we may consider there exists $z_1,z_2,z_3$ so that $xyz_1z_2z_3$ is a $C_5$ in $G$ and this $C_5$ must have been mapped to line in $C_{5^{r+1}}$. Suppose now it happens for an edge that is not in a $C_5$ and remember these edges come in triplicate. If $xy$ were the edge in the original graph, then consider that they became edges $x'y'$, $x''y''$, $x'''y'''$ in $G$. Now, w.l.o.g., if $h(x')$ and $h(y')$ are mapped at distance $0<c<5^{r}$ in $C_{5^{r+1}}$, then both $h(x'')$ and $h(y'')$, and $h(x''')$ and $h(y''')$ must be mapped to distance $5^r-c \bmod 5^r$ since there are even cycles involving $x'y'$ and $x''y''$, and $x'y'$ and $x'''y'''$, where all other edges were in a $C_5$. But now we consider that there is an even cycle involving $x''y''$ and $x'''y'''$, where all other edges were in a $C_5$, and derive a contradiction.
    
    Now, once we calculate $i$ and subtract it, we can divide by $5^r$ to use $h$ restricted to $X$ to give a homomorphism from $G$ to $C_5$.
    \qed
\end{proof}
\begin{theorem}
\label{thm:graph-hom-main}
    {\sc Graph Homomorphism} is a C123-problem.
\end{theorem}

\section{Locally Constrained Homomorphisms} 
\label{sec:graph-local-homomorphism} 

In this section we consider three locally constrained homomorphism problems and show that all three of them are C23-problems, which become C123-problems after imposing a degree bound.

\problemdef{Locally Bijective Homomorphism}{A graph $G$ and a graph $H$}{Is there a locally bijective homomorphism from $G$ to $H$?}

\problemdef{Locally Injective Homomorphism}{A graph $G$ and a graph $H$}{Is there a locally injective homomorphism from $G$ to $H$?}

\problemdef{Locally Surjective Homomorphism}{A graph $G$ and a graph $H$}{Is there a locally surjective homomorphism from $G$ to $H$?}

\noindent
We will often use the abbreviations {\sc LBHom}, {\sc LSHom} and {\sc LIHom} for the three problems. We would also like to consider the bounded-degree versions of these problems, which we refer to as {\sc Degree-$d$-LBHom}, {\sc Degree-$d$-LSHom} and {\sc Degree-$d$-LIHom}, here we restrict the maximum degree of either $G$ or $H$ to be~$d$; note 
that we require only one of these two graphs to have bounded degree.

\begin{lemma}[\cite{CFHPT15}]
    For each $d$, the three problems {\sc Degree-$d$-LBHom}, {\sc Degree-$d$-LSHom} and {\sc Degree-$d$-LIHom} satisfy C1.
\end{lemma}

\begin{lemma}[\cite{CFHPT15}]
  The three problems  {\sc Degree-$3$-LBHom}, {\sc Degree-$3$-LSHom}, {\sc Degree-$3$-LIHom} satisfy C2.
\end{lemma}

In particular {\sc LBHom}, {\sc LSHom} and {\sc LIHom} remain \NP-complete where $G$ is subcubic and $H$ is $K_4$ \cite{CFHPT15}.

\begin{lemma}
    {\sc Locally Bijective Homomorphism} satisfies C3.
\end{lemma}
\begin{proof}
    We claim {\sc LBHom} is \NP-complete for $r$-subdivisions of subcubic graphs for any integer $r$. Let $X$ be the set of vertices of $G^r$ that also appear in $G$ and let $Z$ be the vertices of $K_4^r$ that also appear in $K_4$. We claim:
        \[G \xrightarrow{\text{B}} K_4 \mbox{ if and only if } G^{r} \xrightarrow{\text{B}} K_4^{r}.\]
    The forward direction is trivial. Let us consider the backward direction, let $h_b$ be a locally bijective homomorphism from $G^r$ to $K_4^r$, we claim $h_b(v) \in Z$ if, and only if $v \in X$. As $h_b$ is locally bijective, degree must be preserved, meaning $h_b(v) \in Z$ if, and only if $v$ has degree 3 therefore by showing all vertices in $X$ must have degree 3 our claim is proven.

    We may assume there exists at least one degree~$3$ vertex $v \in G^r$, let $P=(v,p_1,p_2, \ldots, p_{r-1}, p_{r})$ be an arbitrary path of length $r$ from $v$. Both $v$ and $p_{r}$ must be in $X$ with all intermediate vertices having degree 2. Similarly for any path of length $r$ from $h_b(v)$ as $h_b(v) \in Z$ all intermediate vertices must have degree 2 with that at distance $r$ in $Z$. As $p_1$ must be mapped to some neighbour of $h_b(v)$, both $p_1$ and $h_b(p_1)$ have degree 2. Now as $p_2$ cannot map to $h_b(v)$ it must map to the next vertex on a path away from $h_b(v)$. As both $p_{i-1}$ and $h_b(p_{i-1})$ have degree 2 it follows that $h_b(p_i)$ must have distance $i$ from $h_b(v)$. Inductively it follows that $h_b(p_{r})$ must have distance $r$ from $h_b(v)$ implying $h_b(p_{r}) \in Z$ and therefore $p_{r}$ has degree 3. As this holds for all vertices distance $r$ from a degree 3 vertex it follows that all vertices in $X$ must have degree 3 meaning $h_b$ can be used restricted to $X$ to give a homomorphism from $G$ to $K_4$.
\qed
\end{proof}
\begin{figure}
{\color{black}
\[
\resizebox{!}{0.06cm}{
\xymatrix{
& & & \bullet \ar@{-}[d] \ar@{-}[rr] & & \ar@{-}[d] \ar@{-}[rr] & & \ar@{-}[d] \ar@{-}[rr] & & \ar@{-}[d] \ar@{-}[rr] & & \ar@{-}[d] \ar@{-}[rr] & & \ar@{-}[d] \ar@{-}[rr] & & \ar@{-}[d] & & \\
& & & \bullet \ar@{-}[dr] & & \bullet \ar@{-}[dl] \ar@{-}[dr] & & \bullet  \ar@{-}[dl] \ar@{-}[dr] & & \bullet  \ar@{-}[dl] \ar@{-}[dr] & & \bullet  \ar@{-}[dl] \ar@{-}[dr] & & \bullet  \ar@{-}[dl] \ar@{-}[dr] & & \bullet \ar@{-}[dl] \\
& & & & \bullet & & \bullet & & \bullet & & \bullet & & \bullet & & \bullet & & \\
\bullet \ar@{-}[dd] \ar@{-}[r] & \bullet \ar@{-}[dr] & & & & & & & & & & & & & & & & \bullet \ar@{-}[r] \ar@{-}[dl] & \bullet \ar@{-}[dd] \\
& & \bullet \ar@{-}[uurr] & & & & & & & & & & & & & & \bullet \ar@{-}[uull] \\
\bullet \ar@{-}[r] & \bullet \ar@{-}[ur] & & & & & & & & & & & & & & & & \bullet \ar@{-}[r] \ar@{-}[ul] & \bullet  \\
& & \bullet \ar@{-}[ul] \ar@{-}[uuuurrrr] & & & & & & & & & & & & & & \bullet \ar@{-}[ur] \ar@{-}[uuuullll] \\
\bullet \ar@{-}[uu] \ar@{-}[r] & \bullet \ar@{-}[ur] & & & & & & & & & & & & & & & & \bullet \ar@{-}[r] \ar@{-}[ul] & \bullet \ar@{-}[uu] \\
& & \bullet \ar@{-}[ul] \ar@{-}[uuuuuurrrrrr]  & & & & & & & & & & & & & & \bullet \ar@{-}[ur] \ar@{-}[uuuuuullllll]  \\
\bullet \ar@{-}[uu] \ar@{-}[r] & \bullet \ar@{-}[ur] & & & & & & & & & & & & & & & & \bullet \ar@{-}[r] \ar@{-}[ul] & \bullet \ar@{-}[uu] \\
& & \bullet \ar@{-}[ul]  \ar@{-}[ddddddrrrrrr] & & & & & & & & & & & & & & \bullet \ar@{-}[ur]  \ar@{-}[ddddddllllll] \\
\bullet \ar@{-}[uu] \ar@{-}[r] & \bullet \ar@{-}[ur] & & & & & & & & & & & & & & & & \bullet \ar@{-}[r] \ar@{-}[ul] & \bullet \ar@{-}[uu] \\
& & \bullet \ar@{-}[ul] \ar@{-}[ddddrrrr] & & & & & & & & & & & & & & \bullet \ar@{-}[ur]  \ar@{-}[ddddllll]\\
\bullet \ar@{-}[uu] \ar@{-}[r] & \bullet \ar@{-}[ur] & & & & & & & & & & & & & & & & \bullet \ar@{-}[r] \ar@{-}[ul] & \bullet \ar@{-}[uu] \\
& & \bullet \ar@{-}[ul] \ar@{-}[ddrr] & & & & & & & & & & & & & & \bullet \ar@{-}[ur] \ar@{-}[ddll]\\
\bullet \ar@{-}[uu] \ar@{-}[r] & \bullet \ar@{-}[ur] & & & & & & & & & & & & & & & & \bullet \ar@{-}[r] \ar@{-}[ul] & \bullet \ar@{-}[uu] \\
& & & & \bullet & & \bullet & & \bullet & & \bullet & & \bullet & & \bullet & & \\
& & & \bullet \ar@{-}[ur] & & \bullet \ar@{-}[ul] \ar@{-}[ur] & & \bullet  \ar@{-}[ul] \ar@{-}[ur] & & \bullet  \ar@{-}[ul] \ar@{-}[ur] & & \bullet  \ar@{-}[ul] \ar@{-}[ur] & & \bullet  \ar@{-}[ul] \ar@{-}[ur] & & \bullet \ar@{-}[ul] \\
& & & \bullet \ar@{-}[u] \ar@{-}[rr] & & \ar@{-}[u] \ar@{-}[rr] & & \ar@{-}[u] \ar@{-}[rr] & & \ar@{-}[u] \ar@{-}[rr] & & \ar@{-}[u] \ar@{-}[rr] & & \ar@{-}[u] \ar@{-}[rr] & & \ar@{-}[u]  & & \\
}
}
\]
\caption{Self-reduction at the start of the proof of Lemma~\ref{lem:graph-hom-weak-C3-bis}. This is what the $4$-cycle $C_4$ would be reduced to.}
\label{fig:new-figure-proof-bis}
}
\end{figure}
\begin{lemma}
    {\sc Locally Surjective Homomorphism} satisfies C3.
    \label{lem:local_s-C3}
\end{lemma}
\begin{proof}
\medskip
\noindent
As with the bijective case let $X$ be the set of vertices of $G^r$ that also appear in $G$ and let $Z$ be the vertices of $K_4^r$ that also appear in $K_4$. Now we claim:

\medskip
\noindent
$G \xrightarrow{\text{S}} K_4 \mbox{ if and only if } G^{r} \xrightarrow{\text{S}} K_4^{r}$.

Again the forward direction is trivial and let us consider the backward direction, let $h_s$ be a locally surjective homomorphism from $G^r$ to $ K_4^r$, we claim $h_s(v) \in Z$ if and only if $v \in X$. In the surjective case, the degree of a vertex of $v \in G^r$ can be greater than that of $h_s(v)$, although the inverse cannot be true thus implying that only vertices in $X$, specifically those with degree 3, can be mapped to a vertex in $Z$.

Now consider the other direction showing if $v \in X$ then $h_b(v) \in Z$, assume for contradiction, $v$ were mapped to some vertex not in $Z$. As all vertices in $G^r$ with degree 3 must be in $X$, the length of the shortest path from $v$ to the closest degree 3 vertex must be $\geq r$. Let $x$ be this degree 3 vertex and $P=(v, p_1, p_2,\ldots,x)$ be the path from $v$ to $x$. In addition, as $h_s(v) \notin Z$ there must be exactly two vertices, $z, z' \in Z$, with distance $< r$ from $h_s(v)$ with every path of length $r$ from $h_s(v)$ containing one of these. We claim the shortest path from $v$ to $x$ must map to a path from $h_s(v)$ via $z$ or $z’$, note this leads to a contradiction as a degree 2 vertex cannot be mapped to a vertex with degree 3.

First consider the vertex $p_2$, necessarily $h_s(p_2)$ is some neighbour of $h_s(v)$. If $p_2$ maps to either $z$ or $z’$ our claim is proven, else, $h_b(p_2)$ must have degree 2. Each neighbour of $h_s(p_2)$ must be mapped to at least one neighbour of $p_2$, as $p_2$ has degree 2, $p_3$ must map to a vertex distance 2 from $v$ given it cannot map to $h_b(v)$. It then follows inductively that while $i<r$ and $h_s(p_i) \notin \{z,z’\}$, both $p_i$ and $h_b(p_i)$ have degree 2, meaning $h_s(p_{i+1})$ must have distance $i$ from $h_s(v)$. As every path of length $r$ from $h_s(v)$ must go via either $v$ or $v’$ we once again have a contradiction. As $h_s(v) \in Z$ if and only if $v \in X$ we can again use $h_s$ restricted to $X$ to give a homomorphism from $G$ to $K_4$.
\qed
\end{proof}

\begin{lemma}
    {\sc Locally Injective Homomorphism} satisfies C3.
    \label{lem:local_i-C3}
\end{lemma}
\begin{proof}
\medskip
\noindent
Similarly to {\sc LBHom} and {\sc LSHom}, we claim:

\medskip
\noindent
$G \xrightarrow{\text{I}} K_4 \mbox{ if and only if } G^{r} \xrightarrow{\text{I}} K_4^{r}$.

Again, the forward direction is trivial. For the backward direction, let $h_i$ be a locally injective homomorphism from $G^r$ to $ K_4^r$, this means the degree of $v$ can be less than that of $h_i(v)$ but the inverse cannot be true. We will again show that $h_i(v) \in Z$ if and only if $v \in X$. To aid in this we make a second claim, given two vertices distance $r$ from each other, call these $v$ and $v^{r}$, then $v^{r}$ maps to some vertex in $Z$ if and only if $v$ maps to a vertex in $Z$.

Let $P = (v,p_1,p_2,\ldots ,p_{r-1}, v^{r})$ be the shortest path between $v$ and $v^{r}$. First consider where $h_i(v) \in Z$, given any path of length $r$ from $h_i(v)$, all intermediate vertices must have degree 2 and the vertex at distance $r$ must also be in $Z$. Showing $P$ must map to a path of length $r$ from $h_i(v)$ therefore implies $h_i(v^{r}) \in Z$.

$p_2$ must map to some neighbour of $h_i(v)$, as all neighbours of $h_i(v)$ have degree 2, $p_2$ must have degree $\leq 2$ and as it lies on a path it must have degree 2. Then as no two neighbours of $p_2$ can map to the same neighbour of $h_i(p_2)$, the vertex $p_3$ must map to a vertex distance 2 from $v$. Inductively it follows, where $i<r$, if $h_i(p_i)$ has distance $i-1$ from $h_i(v)$, given $h_i(p_i)$ and therefore also $p_i$ must have degree 2, $h_i(p_{i+1})$ has distance $i$ from $h_i(v)$. Thus $P$ maps to a path of length $r$ from $h_i(v)$, implying $h_i(v^{r}) \in Z$.

Similarly, if $h_i(v) \notin Z$ a path of length $r$ from $h_i(v)$, is made up of two paths with length $< r$. One from $h_i(v)$ to some vertex $z_j \in Z$ and one away from $z_j$, thus any vertex at distance $r$ cannot be in $Z$. Given that both paths contain only degree 2 vertices, the reasoning from the previous case holds. Say $z_j$ has distance $l$ from $h_i(v)$, the first $l$ vertices of our path must map a path of length $l$ from $h_i(v)$ with the remaining path of length $r-l$ mapping to a path of length $r-l$ from $v_j$. This therefore proves our claim.

Given this holds for any two vertices in $G^r$ distance $r$ from one another, this must also hold for any multiple of $r$. If two vertices $v, w$ have distance $k \cdot r$ from each other, for some integer $k$, then $h_i(v) \in Z$ if, and only if $h_i(w) \in Z$. We can now prove our main claim, $h_i(v) \in Z$ if and only if $v \in Z$.

Say some vertex $v \in X$ were mapped to some $h_i(v) \notin Z$, we may assume $v$ has degree $\leq 2$ as it is mapped to a vertex with degree 2. As we assume there exists at least one degree 3 vertex in $G^r$, let $x$ be that closest to $v$. As both $v$ and $x$ are in $X$ they must have distance some multiple of $r$. This implies $h_i(x) \notin Z$, however, this leads to a contradiction as a degree 3 vertex cannot be mapped to a vertex with degree 2.

In the other direction, assume some vertex $v \notin X$ were mapped to some vertex in $Z$, again let $x$ be the closest degree 3 vertex to $v$. In addition, on the shortest path from $v$ to $x$, let $v'$ be the vertex closest to $x$ such that its distance from $v$ is some multiple of $r$. As $h_i(v') \in Z$ and the distance from $v’$ to $x$ is less than $r$, $h_i(x)$ must lay on the path between two vertices in $Z$. As previously this path from $v'$ to $x$ must map to a path of the same length from $h_i(v')$ to $h_i(x)$. This means $x$ cannot map to a vertex in $Z$ which again leads to a contradiction as $x$ has degree 3 and $h_i(x)$ has degree~$2$.
\qed
\end{proof}

\begin{theorem}
    {\sc Locally Bijective Homomorphism}, {\sc Locally Surjective Homomorphism} and {\sc Locally Injective Homomorphism} are C23-problems.
\end{theorem}

\begin{theorem}
   {\sc Degree-$3$-LBHom}, {\sc Degree-$3$-LSHom} and {\sc Degree-$3$-LIHom} are C123-problems.
\end{theorem}

\section{Sequential $3$-Colouring Construction Game, QCSP$(K_3)$}\label{s-seq}

We prove that QCSP$(K_3)$ remains PSPACE-complete for graphs of bounded pathwidth. Afterwards, we do the same for the {\sc Sequential $3$-Colouring Construction Game}. 
We first formally define the former problem.

\problemdef{QCSP$(K_3)$}{A sentence $\phi$ of the form $Q_1 x_1 Q_2 x_2 \ldots Q_n x_n \ \Phi$, where $Q_i \in \{\forall,\exists\}$ and $\Phi$ is a conjunction of atoms involving the edge relation $E$.}{Is $\phi$ true on $K_3$?}

\noindent
{\sc QCSP$(K_3)$} is sometimes known as \emph{Quantified $3$-Colouring} and as highlighted previously closely relates to the {\sc Sequential 3-Colouring Construction Game} proposed by Bodlaender in \cite{Bo91}. The two key differences between the problems lie in the requirement of strict alternation in players and each player must assign a colour not previously assigned to a neighbour, in particular, we will show the hardness on bounded pathwidth is preserved between problems.

\begin{theorem}\label{t:bounded-3-col-pspace}
    {\sc QCSP$(K_3)$} is PSPACE-complete for graphs of bounded pathwidth.
\end{theorem}
\begin{proof}
We will reduce from the PSPACE-complete problem Quantified Boolean Formulas (QBF) which was shown by Atserias et al. \cite{AO14} to remain hard where the input is in CNF and the path-width of the primal constraint graph is constant.

Let $\phi = Q_1 x_1 Q_2 x_2 \ldots Q_n x_n \Phi$ be an instance of QBF where $\Phi$ is a CNF formula. The 
{\it primal graph} of $\Phi$, $G(\Phi)$, contains a vertex for each variable of $\Phi$ and an edge where two variables appear in a clause together. Let the pathwidth of $G(\Phi)$ be a constant $w$. Each clause may then have length at most $w+1$, as every clause must be contained as a clique in some bag.

To reduce from QBF to {\sc QCSP$(K_3)$} we will construct an intermediate instance of {\sc Quantified Not-All-Equal 3-Satisfiability} (QNAE3SAT), $\phi'$. Each clause $C_i$ of $\Phi$ is replaced by a set of \textit{NAE} relations $C'$ in $\Phi$. Introducing a constant False it is clear that $\mathit{NAE} (C_i, F)$ is satisfied if, and only if $C_i$ is satisfied. Now it remains to ensure that there are three variables in each clause. Given $C_i$ of $\Phi$ contains literals $l_1,\ldots,l_k$, $k \leq w$ we create $k-2$ new existentially quantified variables $q_{i,j}$, $1 \leq j \leq k-2$. We can now define $C'$ as follows, if $k = 1$ let $C'=\mathit{NAE}(l_1,F,F)$, where $k=2$ let $C'=\mathit{NAE}(l_1,l_2,F)$, otherwise where $k \geq 3$, $C' = \mathit{NAE}(l_1,l_2,q_1), \mathit{NAE}(\overline{q_1},l_3,q_2)$, \ldots, $\mathit{NAE}(\overline{q_{k-3}},l_{k-1},q_{k-2}), \mathit{NAE}(\overline{q_{k-2}},l_k,F)$.

We now claim $C'_i$ is satisfied by an assignment of variables if, and only if $C_i$ is satisfied by this same assignment, thus given a winning strategy for Existential in $\phi$ which satisfies $C_i$, this also gives a winning strategy for Existential in $\phi'$ satisfying all clauses in $C'_i$. Take some winning strategy for Existential in $\phi$, at least one literal in $C_i$ must be evaluated as True. If $l_j$ is the first such literal, $l_q \forall 1 \leq q < j$ are evaluated as False. If $j \in \{1,2\}$, variables $q_1\ldots q_{k-2}$ can be assigned False which satisfies all later clauses as each clause contains a positive and negative appearance of a variable. If $j \geq 3$ we can assign $q_1$,\ldots, $q_{j-2}$ True and $q_{j-1}$,\ldots, $q_{k-2}$ False, $C'_{i,1}$ is satisfied as neither $l_1$ nor $l_2$ are assigned True and all other clauses contain either $F$ or a $q$ variable and a negated $q$ variable. If $j < k$ then $C'_{i,j-1}$ is satisfied as $l_j$ is assigned True with $\overline{q_{j-2}}$ and $q_{j-1}$ evaluated to 0. $C'_{i,k-1}$ is satisfied as $\overline{q_{k-2}}$ is assigned True and $F$ must be evaluated to False. 

Now assuming all variables clauses in $C'$ are satisfied we claim at least one literal must be assigned True. Assume otherwise, $q_1$ must be assigned True to satisfy $C'_{i,1}$, now $q_2$ must be assigned True to satisfy $C'_{i,2}$, it follows that all $q$ variables must be assigned True. However, this means $C'_{i,k-1}$ cannot be satisfied, thus leading to a contradiction.

From $\phi'$ we construct an instance of {\sc QCSP$(K_3)$}, $\phi''$. We refer to the variable of $\phi''$ as vertices and the vertices of $K_3$ as $\{1,2,3\}$. With the exception of $F$ every variable $x_i$ in $\phi'$ is replaced by a path of vertices $x_i$, $y_i$, $z_i$ in $\phi''$. We introduce a new vertex $W$ which is made adjacent to $F$ alongside all $x$ and $y$ variables. For each clause $C'_p \in \Phi'$ we also introduce a $K_3$ with each vertex corresponding to a literal of the clause. Consider where the first literal of $C'_p$ is a positive appearance of $v_i \in \phi'$ then the first vertex of $K_3$ will be adjacent to $x_i \in \phi''$. If this literal were $\overline{x_i}$ then the same vertex in $K_3$ would be adjacent to $y_i$. 

All variables in $\phi''$ are existentially quantified except for $z_i$ which follows the quantification of $v_i \phi'$ and the linear ordering of the vertices begins with $W$ and $F$, followed by $z$ vertices which follow the ordering of $x_i \in \phi$, the ordering of the remaining variables does not matter as they share the same quantification. The quantifier prefix is therefore $\exists \exists Q_1 Q_2 \ldots Q_n \exists^*$.

\textcolor{black}{We defer until Lemma~\ref{lem:deferred} the proof that $\phi''$ has bounded pathwidth.}

$(\phi' \to \phi'')$.
Suppose $\phi'$ is a positive instance, that is we can define a winning strategy for Existential that wins over any strategy of Universal, we can map this to a winning strategy for Existential in $\phi''$. Without loss of generality, $W$ can be coloured 3 and $F$ coloured 1, $x$ and $y$ variables must therefore be coloured either 1 or 2. We can now map assignments in $\phi'$ to a colouring $\phi''$. If a universal variable $z_i$ of $\phi''$ is coloured 3, $y_i$ can be coloured 1 and $x_i$ coloured 2; otherwise, $x_i$ can be coloured the same as $z_i$ with $y_i$ assigned the single remaining colour available to it. This allows us to map a colouring of the variables $x_i, y_j, z_i$ associated with a universal vertex $z_i$ to a strategy played by Universal in $\phi'$.

If Universal assigns $x_i$ the colour 1 in $\phi''$ we map this to a strategy such that $x_i$ is assigned False in $\phi$ by Universal, similarly if $x_i$ coloured 2 in $\phi''$ we map to a strategy where $x_i$ is assigned True. Given this assignment of universal variables, we colour the remaining existential variables according to the winning existential strategy in $\phi'$ mapping the strategy to colours as done with the universal variables. Finally, we need to show we can colour each of the $K_3$s. As the colouring of variables maps to a winning strategy of QNAE3SAT, at most two literals in a clause can have the same assignment meaning at most two vertices in the triangle have a neighbour with the same colour. Thus using all three colours we can colour the triangle.

($\phi'' \to \phi'$.)
Now suppose we have a winning strategy for Existential in $\phi''$, we will again translate this into a winning strategy for Existential in $\phi'$. We can assume $W$ and $F$ are coloured 3 and 1, respectively, meaning variables $x$ and $y$ variables must be coloured 1 or 2. 
For any assignment of universal variables in $\phi'$ we can map this to a strategy in $\phi''$. We can now as before read off the existential variables, in particular our values of $x$. For each clause, all vertices of the $K_3$ must be coloured differently, meaning at most two of the vertices adjacent to the $K_3$ may have the same colour. This means each clause in $\phi'$ must be satisfied giving us a winning strategy.

\qed
\end{proof}

\noindent We remind the reader that the path-width of a formula is equal to the path-width of the primal graph of its quantifier-free part.

\begin{lemma}
    The path-width of 
    $\Phi''$ is at most $9w+2$.
    \label{lem:deferred}
\end{lemma}
\begin{proof} 
    Given the path-width of $\Phi$ is $w$ we claim the path-width of 
    $\Phi''$ is $w'$ where $w' \leq 9w+2$. Let $B$ be a path decomposition for $\Phi$ such that the size of the largest bag is $w+1$. As each clause $C_i$ of $\Phi$ appears as a clique in the primal graph all variables in $C_i$ must be contained in at least one bag together. We say a clause is associated with a given bag if it is the first bag such that all variables in the clause appear together. Where a bag is associated with multiple clauses it can be replaced with a path of duplicates, thus a bag is associated with at most one clause. Notice this increases the number of bags by at most the number of clauses.

    We now define a path decomposition $B'$ for $\Phi''$ with a bag $b'_i \in B'$ for each $b_i \in B$. Where $b_i$ contains a variable $v_j \in \Phi$, then $b_j$ contains $x, y, z \in \Phi''$. In addition, $W$ and $F$ are contained in each bag of $B'$ meaning $\vert b_i' \vert = 3\vert B_i \vert +2$. For a clause $c_i \in \phi$ of size $k$ we have $k-2$ new variables in $\Phi'$ $q_{i,j}$ for $1 \leq j \leq k-2$ thus $3(k-2)$ new vertices in $\Phi''$. As $c_i$ was replaced by $k-1$ clauses in $\Phi'$, we also have $3(k-1)$ new vertices for the $K_3$ dedicated to each clause. This leads to a total of $3(k-2)+3(k-1) =6k-9$ new vertices for clause $C_i$. If we call these new vertices $S_{C_i}$, as $N[S_{C_i}] = S_{c_i} \cup C_i$, $S_i$ can also be contained in the bag associated with $C_i$ with this remaining a path decomposition. Thus if bag $B_j$ is associated with clause $C_i$, $\vert B'_j \vert = 3\vert B_i \vert + 6k - 7$ and as $\vert B_p \vert, k \leq w+1$ for all $B_p \in B$, $\vert B_p' \vert \leq 9w+2$ for all $B_p' \in B$.
\qed
\end{proof}

\begin{theorem}
    {\sc Sequential $3$-Colouring Construction Game} is PSPACE-complete for graphs of bounded pathwidth.
\end{theorem}

\begin{proof}
    We now consider the two additional restrictions of the {\sc Sequential $3$-Colouring Construction Game}. Where the Universal player is unable to colour a vertex the same as a previously coloured neighbour the problem is identical to that of  {\sc QCSP$(K_3)$} and strict alternation can be overcome using dummy variables while preserving yes and no instances.
    
    We will therefore use this to show the reduction 
used to prove Theorem~\ref{t:bounded-3-col-pspace} 
also holds for both problems. In $\phi''$ all universally quantified variables are in the form $z_i$ with a single neighbour $y_i$, as $z_i$ comes before $y_i$ in the linear ordering of vertices. This means the problems of {\sc Sequential $3$-Colouring Construction Game} and {\sc QCSP$(K_3)$} are equivalent on $\phi''$ thus both problems remain hard for graphs of bounded path-width.
\qed
\end{proof}

\section{$\Pi_{2k}$-QCSP$(K_3, \{1,2\},\{1,3\})$}
\label{s-seq-plus}

In this section, we prove that {\sc QCSP$(K_3, \{1,2\},\{1,3\})$} is a C23-problem, but the main result in this section is that its restriction {\sc $\Pi_{2k}$-QCSP$(K_3, \{1,2\},\{1,3\})$} is a C123-problem. We obtain {\sc QCSP$(K_3, \{1,2\},\{1,3\})$} by augmenting {\sc QCSP$(K_3)$} with some unary relations that allow us to restrict existential variables to some subset of the domain.

\problemdef{QCSP$(K_3, \{1,2\},\{1,3\})$}{A sentence $\phi$ of the form $Q_1 x_1 Q_2 x_2 \ldots Q_n x_n \ \Phi$, where $Q_i \in \{\forall,\exists\}$ and $\Phi$ is a conjunction of atoms involving the edge relation $E$ and the unary relations $\{1,2\},\{1,3\}$.}{Is $\phi$ true on $K_3$?}

\noindent
Another variant,
{\sc QCSP$(K_3, \{1,2\},\{1,3\},\{2,3\},\{1\},\{2\},\{3\})$} is also known as \emph{Quantified List $3$-Colouring}. We consider a slight simplification of this problem as we show not all lists are necessary to ensure hardness, however, the hardness extends to the more general case.

In order to occupy our framework, we will consider bounded alternation versions of our problems. 
The problem {\sc $\Pi_{2k}$-QCSP$(K_3, \{1,2\},\{1,3\})$} is the restriction of {\sc QCSP$(K_3, \{1,2\},\{1,3\})$} to inputs in 
$\Pi_{2k}$-form. 
While it may lead to a less natural variant of {\sc Sequential List $3$-Colouring Construction Game} as the original game insists on a strict alternation between the two players. In this case, bounding the number of alternations between the two players would make the game trivial, we instead allow each player to colour multiple vertices in their turn.

\begin{lemma}
    {\sc $\Pi_{2k}$-QCSP$(K_3, \{1,2\},\{1,3\})$} is $\Pi^{\mathbf{P}}_{2k}$-complete for $2r$-subdivisions of subcubic graphs.
    \label{lem:bounded-list-3-col}
\end{lemma}
\begin{proof}
We will reduce from the $\Pi^{\mathbf{P}}_k$-complete problem {\sc $\Pi_k$-Quantified-Not-All-Equal-3-SAT} ($\Pi_k$-QNAE-3-SAT) \cite{Edith04,Martin05a}. Let $\phi = Q_1 x_1 Q_2 x_2 \ldots Q_n x_n \Phi$ be an instance of $\Pi_k$-QNAE-3-SAT where $x_i \in \{0,1\}$, $\Phi = NAE_3(C_1) \lor NAE_3(C_2) \lor \ldots \lor NAE_3(C_m)$ where $C_i = (x_h, x_i, x_j)$ and $x_h,x_i,x_j \in \{x_1,x_2,\ldots,x_n\}$.

For each variable $x_i$ in $\phi$ there are three variables $\exists x_i$, $Q_i z_i$ in $\phi'$ with $L(x_i) = \{1,2\}$ and $L(z_i) = \{1,2,3\}$. Paths of length $2p+1$ are introduced between $z_i$ and $x_i$ with each inner vertex having the list $\{1,2\}$. In addition for each clause $C_p$ we use the three literal clause gadget from \cite{GP14} introducing variables $C_p ,C'_p,C''_p$ where $L(C_p)=\{1,2\}$ and $L(C'_p) = L(C''_p)\{1,2,3\}$. We add the paths of length $2p+1$ between the vertices $x_h$ $C_p$, $x_i$ $C'_p$, $x_j$ $C''_p$, $C_p$ $C'_p$, $C_p$ $C''_p$ with all inner vertices assigned the list $\{1,2\}$. Notice this forces the path to alternate between colours and enforces that given a single endpoint is coloured 1 or 2 the other endpoint cannot be coloured the same. There is also a path between $C'_p$ and $C''_p$ with inner vertices assigned the list $\{1,3\}$, which has the same impact where a vertex is coloured 1 or 3.

The ordering of $\phi'$ follows that of $\phi$ for variables $z_i$. It is then followed by all remaining vertices, as these are all existentially quantified their ordering does not matter. This leads to a prefix $\ Q_1 z_1 Q_2 z_2 \ldots Q_n z_n \exists^*$.

\begin{figure}
    \centering
    \begin{tikzpicture}
        \draw[fill=black] (0,0) circle (3pt); 
        \node at (0,-0.5) {$C_p \{1,2\}$};
        \draw[fill=black] (-1,2) circle (3pt);
        \node at (-1,2.5) {$x_h \{1,2\}$};
        \draw[fill=black] (-3,1) circle (3pt);
        \node at (-3,0.5) {$z_h \{1,2,3\}$};
        
        \draw[fill=black] (2,3) circle (3pt); 
        \node at (3,3) {$C_p' \{1,2,3\}$};
        \draw[fill=black] (2,5) circle (3pt);
        \node at (2,5.5) {$x_i \{1,2\}$};
        \draw[fill=black] (0,4) circle (3pt);
        \node at (0,3.5) {$z_i \{1,2,3\}$};
        
        \draw[fill=black] (4,0) circle (3pt); 
        \node at (4,-0.5) {$C_p'' \{1,2,3\}$};
        \draw[fill=black] (5,2) circle (3pt);
        \node at (5,2.5) {$x_j \{1,2\}$};
        \draw[fill=black] (7,1) circle (3pt);
        \node at (7,0.5) {$z_j \{1,2,3\}$};

        \draw[dashed] (0,0) -- (4,0) node [midway, fill=white] {$\{1,2\}$};
        \draw[dashed] (0,0) -- (2,3) node [midway, fill=white] {$\{1,2\}$};
        \draw[dashed] (4,0) -- (2,3) node [midway, fill=white] {$\{1,3\}$};
        
        \draw[dashed] (0,0) -- (-1,2) node [midway, fill=white] {$\{1,2\}$};
        \draw[dashed] (-1,2) -- (-3,1) node [midway, fill=white] {$\{1,2\}$};
        
        \draw[dashed] (2,3) -- (2,5) node [midway, fill=white] {$\{1,2\}$};
        \draw[dashed] (2,5) -- (0,4) node [midway, fill=white] {$\{1,2\}$};
        
        \draw[dashed] (4,0) -- (5,2) node [midway, fill=white] {$\{1,2\}$};
        \draw[dashed] (5,2) -- (7,1) node [midway, fill=white] {$\{1,2\}$};
    
    \end{tikzpicture}
    \caption{Construction for clause $C_p$. Dashed lines denote paths of length $2p+1$ with the label corresponding to the list of internal vertices.}
    \label{fig:QCSP-K_3-hardness-reduction}
\end{figure}

Suppose $\phi$ is a positive instance, that is we can define a winning strategy for the existential player, we can then also define a winning strategy for Existential in $\phi'$. Given the colouring of universal variables in $\phi'$ we can map this to a strategy for the universal player in $\phi$. Where $z_i$ is coloured 1 or 2, $x_i$ must take the other available colour not used by $z_i$, whereas if $z_i$ is coloured 3 without loss of generality we can assign $x_i$ 1. Given $x_i$ is coloured 1 in $\phi'$ we can map this to a strategy where universal assigns 0 to $x_i$ in $\phi$ similarly given a colouring of 2 in $\phi$ we can map this to an assignment of 1 in $\phi$. We can now evaluate the existential vertices $x_i$ in $\phi'$ according to the strategy of existential for variables $x_i$ in $\phi$, again mapping 0, 1 to 1, 2 respectively.

Given this strategy for colouring variables $x_i$ in $\phi'$ we claim we can colour the remaining existential variables therefore showing there is a winning strategy for existential. First note all existential variables $z_i$ can be coloured 3, next as the colouring of variables $x_i$ equates to a winning strategy for existential in $\phi$ for all clauses $C_i = (x_h, x_i, x_j)$ in $\phi$ at most 2 of the variables can be assigned the same value meaning of $C_p ,C'_p,C''_p$ at most 2 of their respective $x$'s can be given the same colour. The colour of $x_h$ fixes that of $C_p$ such that they are different which also forces both $C'_p$ and $C''_p$ to be coloured differently to $C_p$. At least one variable of $x_i, x_j$ must be coloured differently to $x_h$, say $x_i$, this means $C_p$ must be coloured the same as $x_i$. We can then colour $C'_p$ 1 or 2 whichever is not used by $x_i$ leaving $C''_p$ to be coloured 3.

Now given Existential has a winning strategy in $\phi'$ we will construct a winning strategy for her in $\phi$. Now given an assignment of universal variables in $\phi$ we can map this to a colouring of the universal vertices in $\phi'$, given a variable $x_i$ is assigned 0 in $\phi$ we can map this to a strategy of Universal such that $z_i$ is coloured 2, similarly if $x_i$ is assigned 1 we map to $z_i$ coloured 1. Notice as $x_i$ cannot be coloured the same as $z_i$, this is the same strategy such that an assignment of $x_i$ to 0 in $\phi$ maps to one where $x_i$ is coloured 1 $\phi'$. We can now, as previously, evaluate the existential variables of $\phi$ according to the strategy of Existential for variables $x_i$ in $\phi'$, again mapping 1, 2 to 0, 1 respectively.

We now claim every clause of $\phi$ must be satisfied by this assignment. Assume otherwise, let $c_p$ be the first clause such that $NAE_3(C_p)$ is not satisfied, note as this is a winning strategy for Existential in $\phi'$ we must have some valid colouring for $C_p ,C'_p,C''_p$ such that $x_h$, $x_i$, $x_j$ are all equal. Say these were all coloured 1, $C_p$ must be coloured 2, meaning neither $C'_p$ nor $C''_p$ can be coloured 2. These vertices also cannot be coloured 1 due to $x_i, x_j$ leaving only colour 3 for both, however, at most one can be coloured 3 as there is an odd path between them. The same reasoning also holds where $x_h$, $x_i$, $x_j$ are coloured 2, thus a contradiction.

In this construction the only vertices with degree $> 3$ are vertices $x_i$ where $x_i \in \phi$ appears in $> 2$ clauses in this case for each appearance of $x_i$ in $C_j$ we create a new vertex $x_{i,j}$ with an even path to the previous appearance of this variable and an even path to the next. Inner vertices can be coloured $\{1,2\}$ and as this is an even path $x_{i,j} = x_i$ for all $j$.
\qed
\end{proof}
\noindent The proof of Lemma~\ref{lem:bounded-list-3-col} also furnishes the following, which we note in passing.
\begin{lemma}\label{lem:unbounded-list-3-col}
    {\sc QCSP$(K_3, \{1,2\},\{1,3\})$} is $\Pi_2^{\textrm{P}}$-complete for $2r$-subdivisions of subcubic graphs.
\end{lemma}

\begin{lemma}[\cite{Chen04}]
    {\sc $\Pi_{2k}$-QCSP$(K_3, \{1,2\},\{1,3\})$} satisfies C1.
\end{lemma}

\begin{theorem}
    {\sc $\Pi_{2k}$-QCSP$(K_3, \{1,2\},\{1,3\})$} is a C123-problem.
    \label{thm:C123-bounded-list-col}
\end{theorem}

\noindent
Unlike the bounded alternation case, {\sc QCSP$(K_3,\{1,2\},\{1,3\})$} does not satisfy C1 with the hardness under bounded path-width following directly from {\sc QCSP$(K_3)$}.

\begin{theorem}
    {\sc $\QCSP(K_3, \{1,2\},\{1,3\})$} is a C23-problem.
\end{theorem}

\section{{\color{black}{Forbidding labelled graphs}}}
\label{sec:new-for-journal}

{\color{black}
It has been observed by Hubie Chen\footnote{Private communication, 2024.} that the objects one should be forbidding from inputs to a problem such as $\QCSP(A)$ are not just graphs, but graphs with vertices of two types (universal and existential), as well as a relative ordering on those that represent the quantifier order. In this section, we explore some simple cases where the vertices must have one of the two types, but the order is left undefined. Essentially, we forbid the structure for any of the orderings. Let $H$ be a graph in which all the vertices are \emph{labelled} by either $\exists$ or $\forall$ (we may think of this as a $2$-coloured graph)\footnote{We resist using a different notation for these $2$-labelled graphs. Whether $H$ represents a simple graph or a $2$-labelled graph will be clear from the context.}. An input $\phi$ (with quantifier-free part $\Phi$) for $\QCSP(A)$ is considered $H$-subgraph-free if the graph $G(\Phi)$, having been $2$-labelled according to which vertices are universal and which are existential, omits $H$ as a ($2$-labelled) subgraph. Let us remark on some simple general results. If $H$ is a subgraph of $H'$ then hardness for the $H$-subgraph-free case translates immediately to hardness for the $H'$-subgraph-free case.

\begin{corollary}
    Let $H$ be a $2$-labelled graph. Then
    {\sc $\Pi_{2k}$-QCSP$(K_3, \{1,2\},\{1,3\})$} is $\Pi^{\mathbf{P}}_{2k}$-complete for $H$-subgraph-free graphs if $H$ contains some connected component, $C$, such that any of the following apply:
    \begin{enumerate}
        \item The underlying graph of $C$ is not in $\mathcal{S}$
        \item $C$ contains some universal vertex with degree $>1$
        \item $C$ contains some degree $3$ vertex and some universal vertex 
        \item $C$ contains at least two universal vertices
    \end{enumerate}
    \label{cor:simple}
\end{corollary}
\begin{proof}
    The proofs for all of these lie in previous parts of the paper.  The first part is the principal result as well as a clear corollary of Theorem~\ref{thm:C123-bounded-list-col}.

    The remaining parts follow from the hardness proof of Lemma~\ref{lem:bounded-list-3-col}.

    For the second part, one may note that the hardness proof only uses universal vertices of degree $1$.

    For the third part, one can note that we can make a path from a universal variable to a vertex of degree at least $3$ as long as we care, while not sacrificing hardness.

    The fourth part follows similarly to the third part, as we can make the paths from the universal variables arbitrarily long, we can effectively make universal variables arbitrarily far apart.
\end{proof}
We will now make a further simplification, to study only $\Pi_2$ instances. Let us recall that  {\sc $\Pi_{2}$-QCSP$(K_3, \{1,2\},\{1,3\})$} is $\Pi_2^{\textrm{P}}$-complete in general. We will consider this problem on $H$-subgraph-free instances, where $H$ is a $2$-labelled graph of small size. We will use the following notation. If we assume $H$ to be just a graph with vertex set $V(H)=\{1,\ldots,n\}$, then $H(\alpha)$, where $\alpha \in \{\forall,\exists\}^n$ indicates the graph annotated with vertex $i$ labelled by $\alpha(i)$. For example, $P_3(\exists\forall\exists)$ indicates the path on three vertices with the middle vertex being universal while the end vertices are existential. We use, e.g., $\exists^3$ as an abbreviation for $\exists\exists\exists$.


\subsection{$H$ contains some component of size at most $2$}

\begin{lemma}\label{lem:h+p2}
	Let $H'$ be some $2$-labelled graph, if {\sc $\Pi_{2}$-$\QCSP(K_3, \{1,2\},\{1,3\})$} is:
    \begin{itemize}
        \item $\Pi_2^{\textrm{P}}$-complete for $H'$-subgraph-free graphs then it is also $\Pi_2^{\textrm{P}}$-complete for $(H'+P_2(\exists^2))$-subgraph-free graphs;
        \item \NP-complete for $H'$-subgraph-free graphs then it is also \NP-complete for $(H'+P_2(\exists^2))$-subgraph-free graphs;
        \item in P for $H'$-subgraph-free graphs then it is also in P for $(H'+P_2(\exists^2))$-subgraph-free graphs.
    \end{itemize}
\end{lemma}
\begin{proof}
Note that hardness translates immediately. We will make an argument that the same holds for the positive algorithm (be that in \NP\ or \cP).

	Let $\phi$ be an instance of {\sc $\Pi_{k}$-$\QCSP(K_3, \{1,2\},\{1,3\})$} and let $G$ be the $2$-labelled graph arising from $\phi$. Let $A$ contain the set of universal vertices $G$ and $C^E = \{C^E_1,\ldots, C^E_k\}$ be the connected components of $G \setminus A$. If there is some $C^E_i \in C^E$ such that $G[N[C^E_i]]$ is a no-instance then $G$ is a no-instance. As there is no path of existential vertices between a pair of vertices $u \in C^E_i$ and $v \in C^E_j$ where $C^E_i \neq C^E_j \in C^E$, $G$ is a yes-instance if, and only if, $G[N[C^E_i]]$ is a yes-instance for all $C^E_i \in C^E$.

    First assume $G$ contains some contains some $H'$ but no $H'+P_2(\exists^2)$. As $H'$ has a maximum matching of size at most $\left \lfloor \frac{|H'|+1}{2} \right \rfloor$, every $C^E_{i} \in C^E$ has a maximum matching of size $k \leq \left \lfloor \frac{|H'|+1}{2} \right \rfloor$ and there is a set $X =\{x_i: 1 \leq i \leq k\}$ such that every edge of $C^E_{i}$ has some endpoint in $X$.

    \begin{figure}
        \centering
        \begin{tikzpicture}[every node/.style={circle,thick,draw}]
            \node (x) at (0,3) {$x$};
            
            \node[very thick] (e1) at (-2,1.5) {$e_1$};
            \node[very thick] (e2) at (0,1.5) {$e_2$};
            \node[very thick] (e3) at (2,1.5) {$e_3$};
        
            \node (a0)[very thick] at (-2,0) {$a_1$};
            \node (a1)[very thick] at (-1,0) {$a_2$};
            \node (a2)[very thick] at (1,0) {$a_3$};
            \node (a3)[very thick] at (2,0) {$a_3'$};
        
            \draw (x) -- (e1);
            \draw (x) -- (e2);
            \draw (x) -- (e3);
        
            \draw[very thick] (e1) -- (a0);
            \draw[very thick] (e1) -- (a1);
            \draw[very thick] (e2) -- (a1);
            \draw[very thick] (e2) -- (a2);
            \draw[very thick] (e3) -- (a2);
            \draw[very thick] (e3) -- (a3);
        \end{tikzpicture}
        \caption{No instance as described in Lemma~\ref{lem:h+p2}. The subgraph $G_P$ is drawn in bold with existential vertices $e_1,e_2,e_3 \in P$ each with a pair of universal neighbours in $G_P$. The case $a_1=a'_3$ cannot occur as these vertices are distinct in the path $P$.}        
        \label{fig:no-inst}
    \end{figure}

    We claim if $G[N[C^E_i]]$ contains some path of length $8k+3$ then this is a no-instance. Let $G_P$ denote this $P_{8k+4}$ subgraph and $P= \{p_1,\ldots,p_{8k+4}\}$ denote the vertices of $V(G) \cap V(G_P)$. If $G$  contains some pair of adjacent universal vertices this is a no-instance, that is we assume $A$ is an independent set. This implies that if $G_P \setminus A$ contains $r$ connected components then $P$ contains at most $r+1$ universal vertices. 
    
    If $G_P$ contains some adjacent pair of existential vertices, then at least one of these is in $X$. There are at most $k$ connected components of $G_P \setminus A$ with size at least $2$ (one for each vertex in $X$). Further, as each vertex of $X \cap P$, has at most $2$ neighbours in $G_P$, at most $3k$ existential vertices in $P$ have some existential neighbour in $G_P$. Note all other existential vertices in $P \setminus \{p_1,p_{8k+4}\}$ have a pair of universal neighbours in $G_P$, let $S$ be this set of existential vertices.

    From above $G_P \setminus (A \cup \{p_1,p_{8k+4}\})$ contains at most $3k+|S|$ existential vertices and at most $k+|S|$ connected components and so at most $k+|S|+1$ universal vertices. As $|P| = 8k+4$ it follows that $4k+2|S|+1 \geq 8k+2$ and $|S| \geq 2k+1$ (this uses also the fact that $|S|$ is an integer). Every edge of $C^E_i$ has some endpoint in $X$, so for any $2k+1$ vertices in $V(C^E_i) \setminus X$, at least $3$ have a common endpoint in $X$. This implies that there are vertices $e_1,e_2,e_3 \in S$ such that in $C^E_i$ these have some common neighbour $x \in X$.
    
    We now claim this is a no-instance, in particular we consider the case where $e_2$ has distance $2$ to both $e_1$ and $e_3$ in $G_P$. That is there are vertices $a_1,a_2,a_3,a_4 \in A$ such that $(a_1,e_1,a_2,e_2,a_3,e_3,a_4)$ is a path in $G_p$, see Figure~\ref{fig:no-inst}. Note our arguments will also hold where $e_1,e_2,e_3$ are at distance $>2$ in $G_P$. If $a_1,a_2,a_3,a_4$ are assigned the colours $1,2,3,1$ then $e_1$ must be coloured $3$, $e_2$ coloured $1$ and $e_3$ coloured $2$. Now $x$ has neighbours coloured $1,2,3$ hence this is a no-instance.

    For every $C^E_i \in C^E$,  $G[N[C^E_i]]$ contains no path of length $8k+3$ and so can be solved in polynomial time (by bounded treewidth, following from bounded treedepth, from \cite{Chen04}). We may now assume $G$ contains no $H'$ subgraph in which case the complexity follows on that for $H'$-subgraph-free graphs.
\end{proof}

\begin{lemma}
	Let  $H'$ be some  $2$-labelled graph, if {\sc $\Pi_{2}$-QCSP$(K_3, \{1,2\},\{1,3\})$}  is in \NP\ for  $H'$-subgraph-free graphs then  {\sc $\Pi_{2}$-QCSP$(K_3, \{1,2\},\{1,3\})$} is \NP-complete for $(H'+P_2(\exists\forall))$-subgraph-free graphs.
    \label{lem:h+ae2}
\end{lemma}
\begin{proof}
Note that \NP-hardness follows immediately from any instance that contains only existential variables.

	Let $\phi$ be an instance of {\sc $\Pi_{2}$-QCSP$(K_3, \{1,2\},\{1,3\})$} and let $G$ be the $2$-labelled graph arising from $\phi$ with universal vertices $A \subseteq V(G)$. First assume $G$ contains some $H'$, on the vertices $Z \subseteq V(G)$, but no $H'+P_2(\exists \forall)$. Without loss we may remove any isolated vertices and given any adjacent universal vertices results in a no-instance, we assume every universal vertex is adjacent to some existential vertex. Now given $G$ is $H'+P_2(\exists\forall)$-subgraph-free, $A \subseteq N[Z]$. 
    
    Further if some existential vertex is adjacent to $3$ universal vertices then this is a no-instance, otherwise $|A| \leq 2|V(H')|$. Given $H'$ has constant size we may branch on the possible strategies for $A$ resulting in an arbitrary instance of {\sc CSP$(K_3, \{1,2\},\{1,3\})$} which is \NP-complete. Else where $G$ does not contain some $H'$, the complexity follows that for $H'$-subgraph-free graphs, as {\sc $\Pi_{2}$-QCSP$(K_3, \{1,2\},\{1,3\})$}  is in \NP\ for $H'$-subgraph-free graphs this concludes our proof.
\end{proof}
\noindent The results of this section can be subsumed in the following theorem (see also Figure~\ref{fig:con-comp-2}).

\begin{theorem}
Let $H$ be a graph with each vertex labelled as either universal or existential, such that every component has size at most $2$. If there is some $k\geq 0$ such that $H$ is a subgraph of $k P_2(\exists^2)$, then {\sc $\Pi_{2}$-QCSP$(K_3, \{1,2\},\{1,3\})$} is in P.  If $H$ is a supergraph of $P_2(\exists \forall)$ then {\sc $\Pi_{2}$-QCSP$(K_3, \{1,2\},\{1,3\})$} is $\Pi_2^{\textrm{P}}$-complete. Otherwise, {\sc $\Pi_{2}$-QCSP$(K_3, \{1,2\},\{1,3\})$} is \NP-complete.
\label{thm:each-component-two}
\end{theorem}

\begin{figure}[h]
\color{black}
\begin{center}
\begin{tabular}{|c|c|}
\hline
$H$ is a & complexity \\
\hline
subgraph of $kP_2(\exists^2)$, & P \\
for some $k\geq 0$ & \\
\hline
supergraph of $P_2(\forall^2)$ & $\Pi_2^{\textrm{P}}$-complete \\
\hline
else & \NP-complete \\
\hline
\end{tabular}
\end{center}
\caption{Complexity of {\sc $\Pi_{2}$-QCSP$(K_3, \{1,2\},\{1,3\})$} when forbidding a labelled graph $H$ such that every component has size at most $2$.}
\label{fig:con-comp-2}
\end{figure}

\subsection{The case $P_5(\exists^5)$.}

\begin{lemma}\label{lem:e5}
	{\sc $\Pi_{2}$-QCSP$(K_3, \{1,2\},\{1,3\})$} is in $P$ for $P_5(\exists^5)$-subgraph-free graphs.
\end{lemma}
\begin{proof}
Let $\phi$ be an instance of {\sc $\Pi_{2}$-QCSP$(K_3, \{1,2\},\{1,3\})$} and let $G$ be the $2$-labelled graph arising from $\phi$. Let $A$ contain the set of universal vertices $G$ and $C^E = \{C^E_1,\ldots, C^E_k\}$ be the connected components of $G \setminus A$. If there is some $C^E_i \in C^E$ such that $G[N[C^E_i]]$ is a no-instance then $G$ is a no-instance. Otherwise given there is no path of existential vertices between some $u \in C^E_i$ and some $v \in C^E_j$ for any $C^E_i \neq C^E_j \in C^E$, this is a yes-instance if, and only if $G[N[C^E_i]]$ is a yes-instance for every $C^E_i \in C^E$.

We claim if for some $C^E_i \in C^E$,  $G[N[C^E_i]]$ contains some path of length $19$, then this is a no-instance. Let $G_P$ be this $P_{20}$ subgraph of $G[N[C^E_i]]$ and $P= \{ p_1,\ldots,p_{20}\}$ denote the vertices of $V(G)\cap V(G_P)$. Towards this we claim the maximum matching of $C^E_i$ has size at most $2$. Assume for contradiction, $(x_1,x_2), (y_1,y_2), (z_1,z_2)$ is a matching in $C^E_i$. Given $C^E_i$ is connected, there is some path from $x_1$ to $y_1$. If this path contains some vertex in $V(C^E_i) \setminus \{x_2,y_2\}$ there is some $P_5(\exists^5)$, that is we assume $x_1,y_1$ are adjacent. There is also some path between $x_1$ and $z_1$, likewise this cannot contain some vertex from $V(C^E_i) \setminus \{x_2,z_2\}$. If there is an edge from $z_1$ (symmetrically $z_2$) to some vertex in $\{x_1,x_2,y_1,y_2\}$ then there is some $P_5(\exists^5)$, a contradiction. 

As $C^E_i$ contains a maximum matching of size at most $2$, there are vertices $x,y$ such that every edge of $C^E_i$ has some endpoint in $\{x,y\}$. If $G_P$ contains some pair of adjacent existential vertices at least one of these must be $x$ or $y$, it follows that $G_P \setminus A$ contains at most $2$ components with size at least $2$. Now, $\{x,y\} \cap P$ has at most $2$ neighbours in $G_P$. Thus, $P$ contains at most $6$ existential vertices with some existential neighbour in $G_P$. Note all other existential vertices in $P \setminus \{p_1,p_{8k+4}\}$ have a pair of universal neighbours in $G_P$, let $S$ be this set of existential vertices.

Now $P\setminus \{p_1,p_{8k+4}\}$ contains at most $|S|+6$ existential vertices and $G_P \setminus (A\cap \{p_1,p_{8k+4}\})$ contains at most $|S|+2$ connected components. If $G$ contains some pair of adjacent universal vertices this is a no-instance, that is we assume $A$ is an independent set. As $G_P \setminus (A\cap \{p_1,p_{8k+4}\})$ contains at most $|S|+2$ connected components $P \setminus \{p_1,p_{8k+4}\}$ contains at most $|S|+3$ universal vertices. Given $|P\setminus \{p_1,p_{8k+4}\}| = 18$, it follows $2|S|+9 \geq 18$ that is $|S| \geq 5$. As $C^E_i$ is connected and every edge has some endpoint in $\{x,y\}$, for any $5$ vertices in $V(C^E_i) \setminus \{x,y\}$, at least $3$ have a common endpoint in $\{x,y\}$. Without loss we may assume that $S$ contains vertices $e_1,e_2,e_3$ which are all adjacent to $x$ in $C^E_i$.

Following the approach of Lemma~\ref{lem:h+p2}, we again claim this is a no-instance. Let us again consider the minimal example with respect to number of vertices. That is $e_2$ has distance $2$ to both $e_1$ and $e_3$ in $G_P$ and there are vertices $a_1,a_2,a_3,a_4 \in A$ such that $(a_1,e_1,a_2,e_2,a_3,e_3,a_4)$ is a path in $G_p$, see Figure~\ref{fig:no-inst}. Again our arguments will also hold where $e_1,e_2,e_3$ are at distance $>2$ in $G_P$. If $a_1,a_2,a_3,a_4$ are assigned the colours $1,2,3,1$ then $e_1$ must be coloured $3$, $e_2$ coloured $1$ and $e_3$ coloured $2$. Now $x$ has neighbours coloured $1,2,3$ hence this is a no instance.

We may now assume $G[N[C^E_i]]$ is $P_{20}$-subgraph-free, implying it has treedepth at most $20$. It follows that this can be solved in polynomial time (by bounded treewidth, from \cite{Chen04}).
\end{proof}

\subsection{The case $P_5(\forall\exists^4)$.}

\begin{lemma}\label{lem:ae4}
	{\sc $\Pi_{2}$-$\QCSP(K_3, \{1,2\},\{1,3\})$} is \NP-complete for $P_5(\forall\exists^4)$-subgraph-free graphs.
\end{lemma}
\begin{proof}
	Let $\phi$ be an instance of {\sc $\Pi_{2}$-QCSP$(K_3, \{1,2\},\{1,3\})$} and let $G$ be the $2$-labelled graph arising from $\phi$. Let $A$ contain the set of universal vertices $G$ and $C^E = \{C^E_1,\ldots, C^E_k\}$ be the connected components of $G \setminus A$. We assume $G$ contains no isolated vertices as without loss these can be removed. Consider some $C^E_i \in C^E$, if $C^E_i$ does not contain some path of length $5$ then from Lemma~\ref{lem:e5} we may process this component in polynomial time. Say $C^E_i$ contains some path of length at least $6$, as $G$ is $P_5(\forall\exists^4)$-subgraph-free, $G[N[C^E_i]]$ contains no universal vertices, and $C^E_i$ is a connected component of $G$. Else let $P = \{p_1,p_2,p_3,p_4,p_5\}$ be some $P_5$ in $C^E_i$. If some universal vertex is adjacent to some vertex in $p_i \in P$, then $p_i = p_3$, else $G$ contains $P_5(\forall\exists^4)$. Say $p_3$ has some universal neighbour $a$. If $p_1$ has some neighbour $p_1'\neq p_2 \in G[N[C^E_i]]$ then this results in some $P_5(\forall\exists^4)$, that is $p_1$ has degree $1$. Symmetrically the same holds for $p_5$.
    
    Further if $p_2$ (symmetrically $p_4$) has some neighbour $x \in G[N[C^E_i]]$ it must be existential and given $(x,p_2,p_3,p_4,p_5)$ is a $P_5$, it follows that $x$ has no universal neighbours and every universal vertex in $G[N[C^E_i]]$ must be adjacent to $x_3$. If some vertex has $3$ universal neighbours, then this is a no-instance. We therefore assume each vertex has at most $2$ universal neighbours and so $G[N[C^E_i]]$ contains at most $2$ universal vertices. Branching on each of these gives an instance of {\sc CSP$(K_3, \{1,2\},\{1,3\})$} which is polynomial time solvable on $P_6$-subgraph-free graphs (by bounded treewidth).

	That is we may assume $G$ contains only existential vertices, as {\sc CSP$(K_3, \{1,2\},$ $\{1,3\})$} is \NP-complete and contains no $P_5(\forall\exists^4)$-subgraph {\sc $\Pi_{2}$-QCSP$(K_3, \{1,2\},$ $\{1,3\})$} is \NP-complete for $P_5(\forall\exists^4)$-subgraph-free graphs. 
\end{proof}

\subsection{$H$ has size at most $3$.}

We sum up the results of this section in the following synoptic theorem where $H$ has size at most $3$ (\mbox{cf.} also Figure~\ref{fig:annotated}).

\begin{theorem}
Let $H$ be a graph on at most three vertices with each vertex labelled as either universal or existential. If $H$ is a subgraph of $P_3(\exists^3)$, then {\sc $\Pi_{2}$-QCSP$(K_3, \{1,2\},\{1,3\})$} is in P.  If $H$ is a supergraph of any among $P_2(\forall^2)$, $P_3(\exists\forall\exists)$, or $K_3(\exists^3)$, then {\sc $\Pi_{2}$-QCSP$(K_3, \{1,2\},\{1,3\})$} is $\Pi_2^{\textrm{P}}$-complete. Otherwise, {\sc $\Pi_{2}$-QCSP$(K_3, \{1,2\},\{1,3\})$} is \NP-complete.
\end{theorem}
\begin{proof}
The case in which each component of $H$ is of size at most $2$ comes from Theorem~\ref{thm:each-component-two}.

The case in which the underlying graph is $P_3$ is covered a fortiori by Lemmas~\ref{lem:e5} and \ref{lem:ae4} (as well as Corollary~\ref{cor:simple} for the trivial cases in which there is more than one universal variable, or the universal variable is in the middle).

The case in which the underlying graph is $K_3$ is covered by Corollary~\ref{cor:simple} (and follows from Theorem~\ref{thm:C123-bounded-list-col}).
\end{proof}
\begin{figure}[h]
\color{black}
\begin{center}
\begin{tabular}{|c|c|}
\hline
$H$ is a & complexity \\
\hline
subgraph of $P_3(\exists^3)$ & P \\
\hline
supergraph of $P_2(\forall^2)$ & $\Pi_2^{\textrm{P}}$-complete \\
\hline
supergraph of $P_3(\exists\forall\exists)$ & $\Pi_2^{\textrm{P}}$-complete \\
\hline
supergraph of $K_3(\exists^3)$ & $\Pi_2^{\textrm{P}}$-complete \\
\hline
else & \NP-complete \\
\hline
\end{tabular}
\end{center}
\caption{Complexity of {\sc $\Pi_{2}$-QCSP$(K_3, \{1,2\},\{1,3\})$} when forbidding a labelled graph $H$ of size at most $3$.}
\label{fig:annotated}
\end{figure}

}

\section{Long Edge Disjoint Paths}
\label{sec:long}

The {\it treedepth} of a graph $G$ as the minimum height of a forest~$F$ such that for any pair of adjacent vertices in $G$ one must be an ancestor of the other in $F$. In this section, we prove that {\sc Long Edge Disjoint Paths} is a C23-problem that is polynomial-time solvable for graphs of bounded treedepth but \NP-complete for certain classes of bounded path-width.

The {\sc Edge Disjoint Paths} problem takes as input a graph and $k$ terminal pairs $(s_i,t_i)$. The problem asks if there exists $k$ edge-disjoint (but not necessarily vertex-disjoint) paths connecting each of $(s_1,t_1)$, \ldots, $(s_k,t_k)$. The {\sc Edge Disjoint Paths} problem is known to be a C23-problem \cite{BJMOPPSV}. Consider its variant {\sc Long Edge Disjoint Paths} whose input is as {\sc Edge Disjoint Paths} but we require the yes-instances to be paths of length at least $k$.

\begin{theorem}
    {\sc Long Edge Disjoint Paths} is a C23-problem which is in \cP\ on all classes of bounded treedepth.
    \label{lem:LEDP-1}
\end{theorem}
\begin{proof}
To prove that it satisfies C2 we reduce from {\sc Edge Disjoint Paths} where we first $k$-subdivide each edge (noting $k$ is part of the input of this problem). This ensures that any edge-disjoint paths that connect the terminal pairs are of sufficient length.

    The proof that it satisfies C3 is then exactly the same as that for {\sc Edge Disjoint Paths}, except that we should start with instances in the form of the previous paragraph, built from {\sc Edge Disjoint Paths} by performing a $k$-subdivision at the start. These obtained instances are yes-instances of {\sc Long Edge Disjoint Paths} if, and only if their subdivisions are yes-instances. Note that if we started from arbitrary instances of {\sc Long Edge Disjoint Paths}, we might map no-instances to yes-instances after subdivision.
    
     Finally, on classes of bounded treedepth, which do not contain the $m$-vertex path $P_m$ as a subgraph for some integer~$m$, we may assume that $k\leq m$. Now a brute force approach to exploring for the paths may be undertaken. 
     \qed
\end{proof}

\begin{lemma}
Let $\mathcal{G}$ be a class of graphs of pathwidth at most $p$. For each $k$, $\mathcal{G}^k$ has pathwidth at most $p+2$.
\end{lemma}
\begin{proof}
    Let $G$ be an arbitrary graph in $\mathcal{G}$ and $B$ be a path decomposition for $G$ such that the largest bag has size at most $p+1$. From this, a path decomposition $B'$ can be constructed for $G^k$ such that the largest bag has size at most $p+3$. Every edge $(u,v) \in G$ is replaced by a path $u,p_1, p_2, \cdots, p_k, v$ in $G^k$. Consider the first bag $b \in B$ such that $u$ and $v$ appear together, $b$ can be replaced by a path of $k-1$ bags, $b_1,b_2,\cdots, b_{k-1}$ in $B'$. Each bag contains all the vertices of $b$, and for each pair of adjacent vertices in the path $p_1, p_2, \cdots, p_k$ there will be a bag containing both. In particular $b_i = b \cup \{p_i,p_{i+1}\}$ where $1 \leq i \leq k-1$ as $p_i$ will only be adjacent to $u,v, p_{i-1}, p_{i+1}$ thus making this a valid path decomposition. In addition, where multiple edges first appear in $b$ we can apply this procedure for each path in series. 
\qed        
\end{proof}

\begin{theorem}
     {\sc Long Edge Disjoint Paths} is \NP-complete for graphs of bounded pathwidth and so does not satisfy C1.
\end{theorem}
\begin{proof}
We use the same trick that we used in the proof of Lemma~\ref{lem:LEDP-1} where we reduce from {\sc Edge Disjoint Paths} by first $k$-subdividing each edge (recall $k$ is part of the input of this problem). This is now a correct reduction from {\sc Edge Disjoint Paths} to {\sc Long Edge Disjoint Paths}. Finally, we need that the pathwidth remains bounded and this follows from the previous lemma. \qed  
\end{proof}

\section{Final Remarks}
\label{sec:final-remarks}

In this paper we identified for the first time C123-problems that distinguish between being polynomial-time solvable and hard in the polynomial hierarchy. We also identified several C23-problems, which are hard on bounded pathwidth, including 
  {\sc LBHom}, {\sc LSHom} and {\sc LIHom}. 
Moreover, we proved hardness for bounded path-width for QCSP$(K_3)$ and {\sc Sequential $3$-Colouring Construction Game}. We do not know if the latter two problems satisfy C2 and C3. We leave this for future work.

{\color{black}
We have begun the investigation into QCSP algorithmics under the omission of a labelled graph as a subgraph, which is annotated according to which vertices represent universal and existential variables, respectively. This is a natural framework in which to study QCSPs. We could further have given a relative order for these quantified variables, however we leave this as future work. Our outstanding open question from this section concerns {\sc $\Pi_2$-$\QCSP(K_3, \{1,2\},\{1,3\})$} on $P_m(\exists^m)$-subgraph-free graphs: is this in P? One could as well ask this question for $\Pi_2$-QBF. More generally, if we were looking at (unbounded alternation) $\QCSP(K_3, \{1,2\},\{1,3\})$, our algorithms might need to be more sophisticated.
}

We also gave an example of a problem ({\sc Long Edge Disjoint Paths}) that is a C23-problem that is \NP-complete for some class of graphs of bounded path-width but becomes polynomial-time solvable on all classes of bounded treedepth.
Whether $\QCSP(K_3)$ is another such example is an open question. That is, we do not know if $\QCSP(K_3)$ is polynomial-time for graph classes of bounded treedepth.
This is \textcolor{black}{another} major open problem arising from our work.

\subsubsection*{Acknowledgements.} We are grateful to Mark Siggers for discussions around Theorem~\ref{thm:graph-hom-main}. We thank several anonymous reviewers for useful comments for the final version. The second author is supported by EPSRC grant EP/X03190X/1. The third author is supported by EPSRC Grant EP/X01357X/1.



\end{document}